\documentclass[11pt]{article}
\usepackage{amsfonts,amsthm,amssymb}
\usepackage[pdftex]{graphicx}
\usepackage{color}
\usepackage[fleqn]{amsmath}
\usepackage{hyperref}
\newcommand{\be}{\begin{eqnarray}}
\newcommand{\ee}{\end{eqnarray}}
\newcommand{\bez}{\begin{eqnarray*}}
\newcommand{\eez}{\end{eqnarray*}}
\renewcommand{\d}{\mathrm{d}}
\newcommand{\bd}{\bar{\mathrm{d}}}

\theoremstyle{definition}
\newtheorem{theorem}{Theorem}[section]
\newtheorem{proposition}[theorem]{Proposition}
\newtheorem{corollary}[theorem]{Corollary}
\newtheorem{remark}[theorem]{Remark}
\newtheorem{example}[theorem]{Example}

\numberwithin{theorem}{section}
\numberwithin{equation}{section}

\allowdisplaybreaks

\begin{document}

\title{\bf A vectorial binary Darboux transformation for the first member of 
the negative part of the AKNS hierarchy} 
\author{
 \sc Folkert M\"uller-Hoissen \\ \small 
 Institut für Theoretische Physik, Friedrich-Hund-Platz 1,
 37077 G\"ottingen, Germany \\
 \small  folkert.mueller-hoissen@phys.uni-goettingen.de
}

\date{} 

\maketitle

\begin{abstract}
Using bidifferential calculus, we derive a vectorial binary Darboux transformation for the first member 
of the ``negative'' part of the AKNS hierarchy. A reduction leads to the first ``negative flow" of the 
NLS hierarchy, which in turn is a reduction of a rather simple nonlinear complex PDE in two dimensions, with a leading 
mixed third derivative. This PDE may be regarded as describing geometric dynamics of a complex scalar field 
in one dimension, since it is invariant under coordinate transformations in one of the two independent variables. 
We exploit the correspondingly reduced vectorial binary Darboux transformation to generate multi-soliton 
solutions of the PDE, also with additional rational dependence on the independent variables, and 
on a plane wave background. This includes rogue waves.
\end{abstract}

\section{Introduction}
The main subject of this work is the third-order nonlinear PDE\footnote{The factor $2$ can be 
eliminated by a rescaling of $f$. 
A leading mixed third derivative also appears, for example, in the physically relevant Benjamin-Bona-Mahony \cite{BBM72} and the Joseph-Egri equation \cite{Jose+Egri77}. A mixed third derivative also appears in the Camassa-Holm equation \cite{Cama+Holm71} and in a deformed Hunter-Zheng equation \cite{BDP04}.}
\be
          \left( \frac{f_{xt}}{f} \right)_t + 2 \, (f^\ast f)_x  = 0 \, ,    \label{3rd-order_eq}
\ee
where $f$ is a complex function of two independent real variables $x$ and $t$, and $f^\ast$ is the complex conjugate 
of $f$.  A subscript denotes a partial derivative with respect to one of the independent variables. 
An evident property of (\ref{3rd-order_eq}) is the following. 

\begin{proposition}
\label{prop:gencov}
If $f(x,t)$ solves (\ref{3rd-order_eq}), then also $f(\sigma(x),t)$, with an arbitrary differentiable function $\sigma(x)$.
  \hfill $\qed$
\end{proposition}

This expresses the fact that (\ref{3rd-order_eq}) is invariant under coordinate  transformations $x \mapsto \sigma(x)$ 
in one dimension, and $f$ can be regarded as a scalar. 
A generalization of (\ref{3rd-order_eq}) to higher dimensions is the system
\be
       \frac{\partial}{\partial t} \Big( f^{-1} \, \frac{\partial}{\partial t} \frac{\partial}{\partial x^\mu} f \Big) 
    + 2 \, \frac{\partial}{\partial x^\mu} (f^\ast f)  = 0 \qquad \quad
       \mu =1,\ldots,m \, .  \label{m-dim_3rd-order_eq}
\ee
It behaves as the components of a covector (tensor of type (0,1)) under 
general coordinate transformations in $m$ dimensions, if $f$ is a scalar, also depending on a parameter $t$. 
This system thus defines dynamics of a scalar field on an $m$-dimensional differentiable manifold. Obviously, 
the following holds.

\begin{proposition}
If $f(x,t)$ solves (\ref{3rd-order_eq}), then $f(\sigma(x^1,\ldots,x^m),t)$, with an arbitrary differentiable function 
$\sigma$ of real independent variables $x^\mu$, $\mu=1,\ldots,m$, solves (\ref{m-dim_3rd-order_eq}).  \hfill $\qed$
\end{proposition}

(\ref{m-dim_3rd-order_eq}) can also be written as
\bez
       \frac{\partial}{\partial t} \Big( f^{-1} \, \frac{\partial}{\partial t} \, d f \Big) 
    + 2 \, d (f^\ast f)  = 0 \, ,  
\eez
where $d$ is the exterior derivative on the $m$-dimensional differentiable manifold. 

(\ref{3rd-order_eq}) arises, via a reduction, from the first ``negative flow'' 
of the AKNS hierarchy, which is related to the complex sine-Gordon equation \cite{AFGZ00}. 
In \cite{Kamc+Pavl02} a relation with the sharp line self-induced transparency (SIT) equations 
has been observed, also see \cite{DMH10AKNS,DKMH11acta}. Furthermore, it is connected 
with a 2-component Camassa-Holm equation \cite{CLZ06,AGZ06a}.
Soliton solutions have been found in \cite{JZZ09}, using Hirota's bilinear method. Other methods 
have been applied in particular in \cite{DMH10AKNS} (see Section~5.1 therein), \cite{DKMH11acta}, 
and \cite{Veks12} (see (4.28) therein). 

For \emph{real} $f$, (\ref{3rd-order_eq}) is among the simplest completely integrable PDEs and it has the peculiar 
property mentioned above. In this work, we derive a binary Darboux transformation 
(see, e.g., \cite{Matv+Sall91}) for (more precisely, a reduction of) (\ref{3rd-order_eq}) with \emph{complex} $f$, 
by using a general result of bidifferential calculus, which is 
recalled in Section~\ref{subsec:bDT_in_bdc} in a self-contained way (see, e.g., \cite{DMH20dc} 
for an introduction to bidifferential calculus and references), and demonstrate its use for 
finding soliton solutions. 
We proceed beyond the class of ``simple solitons", which are rational expressions built from 
trigonometric and hyperbolic functions of linear combinations of the independent variables. 
There are also regular solutions which additionally depend \emph{rationally} on the variables $x$ and $t$, as in 
the case of the related Nonlinear Schr\"odinger (NLS) equation (see, in particular, \cite{Pere83,Manas96,AAS09}). 

A binary Darboux transformation for (\ref{3rd-order_eq}) has already been obtained in
\cite{WRH18,Waja+Riaz19,Amja+Khan21}.\footnote{These authors write  (\ref{3rd-order_eq}) in the form
$u_{xt} + 2 u \, \partial_t \partial_x^{-1} (|u|^2) = u$ (hence $x$ and $t$ are exchanged relative to our notation). 
But the term on the right hand side should actually be multiplied by an arbitrary function of $t$, since 
this is the freedom in the definition of an inverse of $\partial_x$. It seems that the authors of 
\cite{Amja+Khan21} wanted to fix the freedom by demanding 
$u \to 0$ as $|x| \to \infty$. But how can this then be reconciled with exact solutions in their 
Section~6, having a non-zero constant background? Throughout our work, there will be no need for 
introducing the inverse of a differential operator.}
The advantages of a \emph{vectorial} binary Darboux transformation, involving 
a ``spectral matrix" instead of a spectral parameter in the underlying linear system, are summarized 
in Remark~\ref{rem:Lax} below. Vectorial binary Darboux transformations for the prominent NLS equation, 
for example, appeared in \cite{Manas96,CDMH16,CMH17}  (also see \cite{Guil+Manas96,CDMH16} for corresponding treatments 
of Davey-Stewartson equations, which possess a reduction to the NLS equation). Such an efficient solution-generating method has been 
obtained by now for many completely integrable equations. It is automatically available if the corresponding equation
possesses a bidifferential calculus formulation, see Section~\ref{sec:bDT}, and many examples can be found 
in particular in \cite{CDMH16} and references cited there. For a different framework, see \cite{Sakh94}. 

In Section~\ref{sec:sym} we recall some useful Lie point symmetries of (\ref{3rd-order_eq}) and traveling 
wave solutions in terms of elementary Jacobi elliptic functions (cf. \cite{AKCBA16} and references cited there 
for corresponding solutions of the ``positive'' part of the NLS hierarchy). 

Section~\ref{sec:sol} presents our version of an $n$-fold binary Darboux transformation for a reduction of (\ref{3rd-order_eq}), 
which we then exploit to find multi-soliton solutions. This includes solitons superposed on a plane wave background. Section~\ref{sec:bDT} presents a derivation of the binary Darboux transformation, more generally 
for the first ``negative flow" of the AKNS hierarchy. 
Finally, Section~\ref{sec:conclusion} contains some concluding remarks.

\section{Some symmetries of the PDE and traveling wave solutions}
\label{sec:sym}
(\ref{3rd-order_eq}) admits the following symmetry transformations:
\begin{itemize}
\item $x \mapsto \sigma(x)$, see Proposition~\ref{prop:gencov}.
\item $t \mapsto \pm t + \alpha$, $\alpha \in \mathbb{R}$.
\item $t \mapsto \pm |\beta| \, t$, $f \mapsto \beta \, f$, $\beta \in \mathbb{C}$, $\beta \neq 0$.
\item $f \mapsto e^{\mathrm{i} \, \varphi_0} \, f$, $\varphi_0 \in \mathbb{R}$. 
\item Complex conjugation of $f$.
\end{itemize}
In the following, solutions will typically be presented modulo these symmetries. 
\vspace{.2cm}

Let us assume that $f$ is real and, in some coordinate $x$, has the form
\bez
          f(x,t) = f(x \pm c \,t) \, ,
\eez
with a real constant $c  > 0$. Then (\ref{3rd-order_eq}) reduces to the ODE\footnote{This is equation 7.7 in \cite{Kamke77}.}
\bez
                \frac{f''}{f} +  \frac{2}{c^2} f^2 = k \, ,
\eez
with a real constant $k$. Exclusively in this section, a prime indicates a 
derivative with respect to the argument of the function $f$. Solutions of this equation are 
provided by the Jacobi elliptic functions $\mathrm{cn}$ and $\mathrm{dn}$ (see \cite{Abra+Steg65}, for example). 
Indeed,
\bez
         f_{\mathrm{cn}} = \sqrt{c} \, \sqrt{m} \; \mathrm{cn}(\frac{1}{\sqrt{c}} (x \pm c \, t) \, | \, m)    
\eez
solves the ODE with $k = (2 m -1)/c$. We note that
\bez
         f_{\mathrm{cn}} =  \sqrt{c} \; \mathrm{sech}( \frac{1}{\sqrt{c}} (x \pm c \, t))  \qquad \mbox{if} \quad m=1  \, .
\eez
We will recover this solitary wave as a single soliton solution in Section~\ref{sec:sol}.
Furthermore,
\bez
         f_{\mathrm{dn}} = \sqrt{c} \; \mathrm{dn}(\frac{1}{\sqrt{c}} (x \pm c \, t) \, | \, m)    
\eez
satisfies the ODE with $k = (2-m)/c$. We note that
\bez
         f_{\mathrm{dn}} =  \sqrt{c} \; \mathrm{sech}( \frac{1}{\sqrt{c}} (x \pm c \, t) )  \qquad \mbox{if} \quad m=1  \, .
\eez

\section{A binary Darboux transformation and soliton solutions} 
\label{sec:sol}
(\ref{3rd-order_eq}) will be treated in the following form,
\be
         a_t = (f^\ast f)_x \, , \qquad f_{xt} + 2 a \, f = 0  \, .   \label{a,f_PDE}
\ee
This system is invariant under a coordinate transformation $x \mapsto x'$ if 
the function $a$ transforms as  $a \mapsto a' = (\partial x/\partial x') \, a$. 

If $f(x,t)$ is allowed to be complex, then (\ref{3rd-order_eq}) is most likely \emph{not} completely integrable 
\cite{Sako22}. But the reduction of (\ref{a,f_PDE}), obtained by restricting the function $a$ to be real, possesses 
a Lax pair (cf. the Remark~\ref{rem:Lax} below) and is thus completely integrable in this sense. Although the first 
of equations (\ref{a,f_PDE}) requires $a_t$ to be real, this does not exclude an imaginary part of $a$. 
As pointed out in \cite{Sako22}, there is thus another reduction of (\ref{a,f_PDE}), where $a = a_1 +  \mathrm{i} \, a_2$, 
with real functions $a_j$, $j=1,2$, $a_2 \neq 0$, $a_{2t}=0$, and this does not pass the Painlev{\'e} test of 
integrability.\footnote{Via the above coordinate invariance, it can then be (at least locally) achieved that 
$a_2$ is a real constant, different from zero.} 

In the following we will only consider the integrable reduction of (\ref{a,f_PDE}), i.e., we will assume that $a$ is \emph{real}.
We next formulate the main result of this work. A derivation and proof is postponed to Section~\ref{sec:bDT}.

\begin{theorem}
\label{thm:f_eq}
Let $a_0,f_0$ be a solution of (\ref{a,f_PDE}) with real $a_0$. Let $n$-component column vectors $\eta_i$, $i=1,2$, 
be solutions of the linear system
\be 
     && \Gamma \, \eta_{1x} = a_0 \, \eta_1 +  f^\ast_{0x} \, \eta_2 \, , \qquad
            \Gamma \, \eta_{2x} = -a_0 \, \eta_2 + f_{0x} \, \eta_1  \, , \label{eta_x} \\
      && \eta_{1t} = - \frac{1}{2} \Gamma \, \eta_1 + f^\ast_0 \, \eta_2   \, , \qquad
       \eta_{2t} = \frac{1}{2} \Gamma \, \eta_2 -  f_0 \, \eta_1  \, , \label{eta_t}
\ee
where $\Gamma$ is an invertible constant $n \times n$ matrix satisfying the spectrum condition 
$\mathrm{spec}(\Gamma) \cap \mathrm{spec}(-\Gamma^\dagger) = \emptyset$. 
Furthermore, let $\Omega$ be an invertible solution of the Lyapunov equation
\be
     \Gamma \, \Omega + \Omega \, \Gamma^\dagger 
              = \eta_1 \eta_1^\dagger + \eta_2 \eta_2^\dagger \, ,  \label{Lyap} 
\ee
where $^\dagger$ denotes Hermitian conjugation (transposition and complex conjugation). 
Then 
\be
         a = a_0 - (\eta_1^\dagger \, \Omega^{-1} \eta_1)_x \, , \qquad   
         f = f_0 - \eta_1^\dagger \, \Omega^{-1} \, \eta_2    \label{BDT_new_solution_f}
\ee
is also a solution of (\ref{a,f_PDE}). As a consequence, $f$ solves (\ref{3rd-order_eq}).    \hfill $\qed$
\end{theorem}

An application of this theorem essentially reduces to solving the linear system for a given solution $a_0, f_0$ 
of (\ref{a,f_PDE}) and a constant $n \times n$ matrix $\Gamma$, since it is well-known that, under 
the stated spectrum condition, (\ref{Lyap}) has a unique solution $\Omega$ and there are concrete 
expressions for it. In order to obtain a more explicit expression for the generated solution $f$ of (\ref{3rd-order_eq}), 
however, one needs to evaluate the inverse of the matrix $\Omega$, which is getting more and 
more difficult with increasing $n$, of course.\footnote{If $n>2$, $\Omega$ can be decomposed into 
block submatrices and the inverse can be computed using Schur complements, which are 
(special) quasideterminants. The latter are used in the alternative iterative approach to a binary 
Darboux transformation in \cite{WRH18,Waja+Riaz19,Amja+Khan21}. }

Without restriction of generality, $\Gamma$ can be restricted to Jordan normal form.
We also note that $\Omega$ in the preceding theorem is Hermitian (also see (\ref{Omega_Hermitian})) and 
consequently $\det(\Omega)$ is real. If $f_0$ is a regular solution of (\ref{3rd-order_eq}) on $\mathbb{R}^2$, 
a solution $f$ generated via the above theorem can only  be singular if $\Omega$ is not invertible somewhere 
on $\mathbb{R}^2$, i.e., if $\det(\Omega)$ has a zero.

\begin{proposition}
\label{prop:linsys_decoupled}
If $a_0, f_0$, with $f_0 \neq 0$ and real $a_0$, solve (\ref{a,f_PDE}), then the linear system (\ref{eta_x}), (\ref{eta_t}) 
is equivalent to
\be
   && \eta_{1tt}  - \frac{f_{0t}^\ast}{f_0^\ast} \, \eta_{1t} -  \Big( \frac{1}{4} \Gamma^2 
           +  \frac{f_{0t}^\ast}{2 f_0^\ast} \Gamma -|f_0|^2 \Big)  \, \eta_1 = 0 \, , \label{eta_tt} \\
   &&  \Gamma \, \eta_{1x} - \frac{f_{0x}^\ast}{f_0^\ast} \eta_{1t} - \Big( a_0 
       + \frac{f_{0x}^\ast}{2 f_0^\ast} \Gamma \Big) \, \eta_1=0  \, ,  \label{eta_x_eta_t} \\
   && \eta_{2} = \frac{1}{f_0^\ast} \, (\eta_{1t} + \frac{1}{2} \Gamma \, \eta_1) \, . \label{eta2}
\ee
\end{proposition}
\begin{proof}
(\ref{eta2}) is the first equation in (\ref{eta_t}), solved for $\eta_2$. The second of (\ref{eta_t}) is 
then equivalent to (\ref{eta_tt}). By use of (\ref{eta2}), the first of (\ref{eta_x}) takes the form 
(\ref{eta_x_eta_t}). As a consequence of these equations, the second of (\ref{eta_x}) is satisfied 
iff  $a_0, f_0$ solve (\ref{a,f_PDE}).
\end{proof}

\begin{remark}
If $f_{0x} = 0$, the second of equations (\ref{a,f_PDE}) requires $a_0=0$. (\ref{eta_x}) then 
restricts $\eta_1$ and $\eta_2$ to not depend on $x$. Since any function independent of $x$ solves 
(\ref{3rd-order_eq}), such an $f_0$ is not a useful ``seed" for the binary Darboux transformation in 
Theorem~\ref{thm:f_eq}. 
\end{remark}

\begin{remark}
\label{rem:Lax}
Writing (\ref{eta_x}) and (\ref{eta_t}) in the form
\bez
        \left( \begin{array}{c} \eta_1 \\ \eta_2 \end{array} \right)_x 
    =  \left( \begin{array}{cc} a_0 \, \Gamma^{-1} & f_{0x}^\ast \, \Gamma^{-1} \\ f_{0x} \, \Gamma^{-1} & - a_0 \, \Gamma^{-1} \end{array} \right) 
       \left( \begin{array}{c} \eta_1 \\ \eta_2 \end{array} \right) 
       \, , \quad
        \left( \begin{array}{c} \eta_1 \\ \eta_2 \end{array} \right)_t 
    =  \left( \begin{array}{cc} - \frac{1}{2} \Gamma & f_0^\ast I_n \\ -f_0 I_n & \frac{1}{2} \Gamma \end{array} \right) 
       \left( \begin{array}{c} \eta_1 \\ \eta_2 \end{array} \right) \, ,
\eez
where $I_n$ is the $n \times n$ identity matrix, 
constitutes a ``Lax pair'' for (\ref{a,f_PDE}) with \emph{real} $a$, since its integrability condition 
is equivalent to $a_0, f_0$ satisfying (\ref{a,f_PDE}) and $a_0^\ast = a_0$. It should be noticed that  
the usual spectral parameter is promoted to a \emph{matrix} $\Gamma$, which is typical for a \emph{vectorial} generalization 
of a (binary) Darboux transformation. It has the effect that there is no need to consider \emph{iterations} of Darboux 
transformations (in contrast to the older approach taken, e.g., in 
\cite{WRH18,Waja+Riaz19,Amja+Khan21}).\footnote{Nevertheless, the classical 
iterative (forward, backward and binary) Darboux transformations can also be formulated in 
bidifferential calculus, see \cite{DMH20dc}.} 
One obtains the result of an $n$-fold elementary binary Darboux transformation right away 
in a single step and efficient matrix methods can be used to elaborate concrete solutions.  
A particular advantage lies in the fact that important classes of completely 
integrable equations are directly obtained by choosing the ``spectral matrix" to be a (non-diagonal) 
Jordan matrix (see \cite{Manas96} for an early example). This concerns, in particular, the rogue wave solutions of the 
nonlinear Schr\"odinger equation (see \cite{CMH17} and references cited there). In some approaches, 
they can only be obtained indirectly, starting from a solution obtained by a multiple application of an 
elementary binary Darboux transformation, each time with a different value of the spectral parameter   
(which then corresponds, in the vectorial setting, to choosing a diagonal spectral matrix with distinct eigenvalues), 
and then taking suitable coincidence limits of the spectral parameter values (and possibly other parameters). 
\end{remark}

\subsection{Some results about the Lyapunov equation}
If $\Gamma = \mathrm{diag}(\gamma_1,\ldots,\gamma_n)$, with $\gamma_i \neq - \gamma_j^\ast$, $i,j=1,\ldots,n$, 
the solution of (\ref{Lyap}) is the Cauchy-like matrix
\bez
       \Omega = \left( \frac{\eta_{1i} \, \eta_{1j}^\ast + \eta_{2i} \, \eta_{2j}^\ast }{\gamma_i + \gamma_j^\ast} \right) \, .
\eez

In the following subsections we will also consider the case where $\Gamma$ is a Jordan matrix. 
Therefore we recall some results from \cite{CMH17}. 
For fixed $n >1$, let $\eta_i = (\eta_{i1}, \eta_{i2},\ldots,\eta_{in})^T$, $i=1,2$.\footnote{A superscript $T$ indicates
the transpose of a matrix.} 
For $1 \leq k \leq n$, let $\Omega_{(k)}$ be the solution of (\ref{Lyap}), where $\Gamma$ is the $k \times k$ lower 
triangular Jordan matrix
\be
       \Gamma_{(k)} = \left( \begin{array}{ccccc} \gamma  & 0      & \cdots & \cdots & 0      \\ 
                                       1  & \gamma  & \ddots  & \ddots & 0      \\ 
                                       0  & \ddots & \ddots & \ddots & \vdots \\
                                   \vdots & \ddots & \ddots & \ddots & 0      \\
                                       0  & \cdots &  0  &   1    & \gamma 
                 \end{array} \right)  \, ,   \label{Gamma_Jordan_block} 
\ee
and $\eta_i$ is replaced by $\eta_{i (k)} := (\eta_{i1}, \eta_{i2},\ldots,\eta_{ik})^T$, $i=1,2$. 

\begin{proposition}
For $1 \leq k \leq n-1$, we have
\bez
     \Omega_{(k+1)} = \left( \begin{array}{cc}
                \Omega_{(k)} & B_{k+1}  \\
                 B_{k+1}^\dagger & \omega_{k+1} \end{array} \right)  \, ,                    
\eez
where
\bez
  && B_{k+1} := K_{(k)}^{-1} \Big( \eta_{1(k)} \, \eta_{1,k+1}^\ast + \eta_{2(k)} \, \eta_{2,k+1}^\ast
                 - \Omega_{(k)} (0,\ldots,0,1)^T \Big) \, , \\
  && \omega_{k+1} := \frac{1}{\kappa} \Big( |\eta_{1,k+1}|^2 + |\eta_{2,k+1}|^2 
                       - 2 \, \mathrm{Re}[(0,\ldots,0,1) B_{k+1}] \Big) \, ,  \qquad 
        \kappa := 2 \, \mathrm{Re}(\gamma) \, , 
\eez
and $K_{(k)}$ is the Jordan matrix $\Gamma_{(k)}$ with $\gamma$ replaced by $\kappa$.    \qed
\end{proposition}

This proposition allows to recursively compute the solution of the Lyapunov equation (\ref{Lyap}) with a Jordan 
matrix $\Gamma_{(n)}$. The inverse of $\Omega_{(n)}$ can also be recursively computed via 
\be
   \Omega_{(k+1)}^{-1} = \left( \begin{array}{cc}
      \Omega_{(k)}^{-1} - S_{\Omega_{(k)}}^{-1} \Omega_{(k)}^{-1} B_{k+1} B_{k+1}^\dagger \Omega_{(k)}^{-1}
             & - S_{\Omega_{(k)}}^{-1} \Omega_{(k)}^{-1} B_{k+1}  \\
      - S_{\Omega_{(k)}}^{-1} B_{k+1}^\dagger \Omega_{(k)}^{-1} &  S_{\Omega_{(k)}}^{-1} 
       \end{array} \right) \, ,   \label{Omega_inverse}
\ee
with the scalar Schur complement 
\bez
    S_{\Omega_{(k)}} = \omega_{k+1} - B_{n+1}^\dagger \Omega_{(k)}^{-1} B_{k+1} \, .
\eez

\begin{example}
\label{ex:n=2Jordan_Lyapunov_sol}
 For $n=2$, we obtain
\bez
  \Omega_{(2)} = \frac{1}{ \kappa } \sum_{i=1}^2 \left( \begin{array}{cc}
                |\eta_{i1}|^2 & \eta_{i1} ( \eta_{i2} - \kappa^{-1} \eta_{i1})^\ast  \\
                 \eta_{i1}^\ast ( \eta_{i2} - \kappa^{-1} \eta_{i1}) &
                    | \eta_{i2} - \kappa^{-1} \eta_{i1}|^2 + \kappa^{-2} |\eta_{i1}|^2 
                     \end{array} \right)  \, .
\eez
Its determinant is
\bez
    \det(\Omega_{(2)}) = \kappa^{-4} (|\eta_{11}|^2 + |\eta_{21}|^2)^2 
            + \kappa^{-2} \, |\det(\eta_1,\eta_2)|^2  \, ,
\eez
where $\det(\eta_1,\eta_2) = \eta_{11} \eta_{22} - \eta_{12} \eta_{21}$. 
\end{example}

\begin{proposition}
Let $\Gamma$ be a lower triangular $n \times n$ Jordan matrix with eigenvalue $\gamma$ and
$\mathrm{Re}(\gamma) \neq 0$. If the first component of one of the vectors $\eta_i$, $i=1,2$, is different 
from zero, then the solution of (\ref{Lyap}) is invertible.  \qed
\end{proposition}

Unfortunately, to our knowledge, a convenient formula for the determinant of the solution of the 
Lyapunov equation with an $n \times n$ Jordan matrix $\Gamma$ 
is only available in the case where its right hand side is a rank one matrix.

\subsection{Zero seed solutions}
If $f_0=0$, we choose $a_0 = - 1/2$. The linear system for $\eta$ then has the solutions
\be
          \eta_1 = \exp\Big( -\frac{1}{2} (\Gamma^{-1} x + \Gamma \, t) \Big) \, v \, , \qquad
          \eta_2 = \exp\Big( \frac{1}{2} (\Gamma^{-1} x + \Gamma \, t) \Big) \, w \, ,    \label{eta_zero_seed}
\ee
where $v,w$ are constant $n$-component column vectors.  The ansatz
\be
       \Omega = e^{ -\frac{1}{2} (\Gamma^{-1} x + \Gamma \, t)} \, X \, 
                           e^{ -\frac{1}{2} (\Gamma^{\dagger-1} x + \Gamma^\dagger \, t)} 
                       + e^{ \frac{1}{2} (\Gamma^{-1} x + \Gamma \, t)} \, Y \, 
                           e^{\frac{1}{2} (\Gamma^{\dagger-1} x + \Gamma^\dagger \, t)} \, ,    \label{Omega_zero_seed}
\ee
with constant $n \times n$ matrices $X,Y$, solves (\ref{Lyap}) if 
\be
            \Gamma \, X + X \Gamma^\dagger = v \, v^\dagger \, , \qquad 
            \Gamma \, Y + Y \Gamma^\dagger = w \, w^\dagger \, .   \label{Lyap_XY}
\ee 
According to Theorem~\ref{thm:f_eq}, 
\be
         f &=&  v^\dagger \, e^{ -\frac{1}{2} (\Gamma^{\dagger-1} x + \Gamma^\dagger \, t)} \,  \Omega^{-1} \, 
                         e^{ \frac{1}{2} (\Gamma^{-1} x + \Gamma \, t)} \, w   \nonumber  \\
           &=&  \frac{1}{\det(\Omega)} \, v^\dagger \, e^{ -\frac{1}{2} (\Gamma^{\dagger-1} x + \Gamma^\dagger \, t)} \, 
                 \mathrm{adj}(\Omega) \, e^{ \frac{1}{2} (\Gamma^{-1} x + \Gamma \, t)} \, w \, ,
                                 \label{zero_seed_multi_soliton}
\ee
where $\mathrm{adj}$ takes the adjugate of a matrix,  
represents (an infinite set of) exact solutions of (\ref{3rd-order_eq}). 
Here we dropped a global minus sign, since $f \mapsto -f$ is a symmetry of (\ref{3rd-order_eq}).  
We also note that $\Gamma \mapsto -\Gamma$ together with $(x,t) \mapsto -(x,t)$ amounts to $f \mapsto -f$. 
Furthermore, $\Gamma \mapsto -\Gamma$ together with exchange of $v$ and $w$ means $f \mapsto -f^\ast$.

\subsubsection{Simple multi-soliton solutions}
These are obtained by choosing $\Gamma = \mathrm{diag}(\gamma_1, \ldots, \gamma_n)$, where 
$\gamma_i^\ast \neq -\gamma_j$, $i,j=1,\ldots,n$. The solutions of the Lyapunov equations (\ref{Lyap_XY}) are 
then the Cauchy-like matrices
\bez
        X = \left( \frac{v_i \, v_j^\ast}{\gamma_i + \gamma_j^\ast} \right) \, , \qquad
        Y = \left( \frac{w_i \, w_j^\ast}{\gamma_i + \gamma_j^\ast} \right) \, .
\eez
Substituting these expressions in (\ref{Omega_zero_seed}), (\ref{zero_seed_multi_soliton}) provides us 
with an exact solution of (\ref{3rd-order_eq}) for any $n \in \mathbb{N}$.

\paragraph{$\mathbf{n=1}$.} 
In this case, (\ref{zero_seed_multi_soliton}) becomes
\bez
          f = 2 \, \mathrm{Re}(\gamma)  \, v^\ast w \, \Big( |v|^2 e^{-(\gamma^{-1} x + \gamma \, t)} 
                             + |w|^2  e^{\gamma^{\ast-1} x + \gamma^\ast \, t} \Big)^{-1} \, . 
\eez
If $\gamma,v,w$ are real, this can be rewritten, up to a global sign, as
\be
          f =  \gamma \,  \mathrm{sech}( \gamma^{-1} x + \gamma \, t + \alpha )  \, ,   \label{1soliton}
\ee
where $\alpha = \ln(|w/v|)$. 
The solitary wave has the form of the bright soliton of the NLS equation and the solitary wave of 
the modified KdV (mKdV) equation. 

\paragraph{$\mathbf{n=2}$.} (\ref{zero_seed_multi_soliton}) yields
\bez
         f &=& \frac{1}{2 \mathrm{Re}(\gamma_1) \, (\gamma_1+\gamma_2^\ast) \, \det(\Omega)} 
                  e^{ - \mathrm{Re}(\gamma_1) (t+x/|\gamma_1|^2) + \mathrm{i} \, \mathrm{Im}(\gamma_2) (t-x/|\gamma_2|^2)}
            \Big( (\gamma_2^\ast-\gamma_1^\ast) |v_1|^2 v_2^\ast w_2     \\
          & & + ( \gamma_1+\gamma_2^\ast) v_2^\ast |w_1|^2  w_2 \, e^{2 \mathrm{Re}(\gamma_1)(t+x/|\gamma_1|^2)} 
                 - 2 \mathrm{Re}(\gamma_1) v_1^\ast w_1 |w_2|^2 \, e^{(\gamma_1+\gamma_2^\ast) t +(\gamma_1^{-1}
                +\gamma_2^{\ast-1}) x} \Big) \\
         & &+ \frac{1}{2 \mathrm{Re}(\gamma_2) \, (\gamma_1^\ast+\gamma_2) \, \det(\Omega)} 
                  e^{ - \mathrm{Re}(\gamma_2) (t+x/|\gamma_2|^2) + \mathrm{i} \, \mathrm{Im}(\gamma_1) (t-x/|\gamma_1|^2)}
            \Big( ( \gamma_1^\ast-\gamma_2^\ast) v_1^\ast |v_2|^2 w_1     \\
          & & + (\gamma_1^\ast+\gamma_2) v_1^\ast w_1 |w_2|^2 \, e^{2 \mathrm{Re}(\gamma_2)(t+x/|\gamma_2|^2)} 
                 - 2 \mathrm{Re}(\gamma_2) v_2^\ast |w_1|^2 w_2 \, e^{(\gamma_1^\ast+\gamma_2) \, t +(\gamma_1^{\ast-1}
                + \gamma_2^{-1}) \, x} \Big) \, , 
\eez
where
\bez
         \det(\Omega) &=& \frac{ e^{ - \mathrm{Re}(\gamma_1 + \gamma_2) \, t - \mathrm{Re}(\gamma_1^{-1}+\gamma_2^{-1}) \, x}}
                            {4 \, \mathrm{Re}(\gamma_1) \, \mathrm{Re}(\gamma_2) \,  |\gamma_1+\gamma_2^\ast|^2}
          \Big( |\gamma_1+\gamma_2^\ast|^2 \, \Big| v_2 w_1 e^{\gamma_1 t + \gamma_1^{-1} x}  
               - v_1 w_2  e^{\gamma_2 t + \gamma_2^{-1} x}\Big|^2 \\
         & &   + |\gamma_1 - \gamma_2|^2 \, \Big| v_1 v_2^\ast + w_1 w_2^\ast \, 
             e^{ ( \gamma_1 + \gamma_2^\ast ) \, t + (\gamma_1^{-1} + \gamma_2^{\ast-1} ) \, x} \Big|^2 \Big) \, ,
\eez
which we were able to express in an explicitly non-negative form.

\begin{proposition}
The 2-soliton solution is regular if $\gamma_1, \gamma_2 \neq 0$, $\gamma_1 \neq \gamma_2$,  $\gamma_1 \neq -\gamma_2^\ast$,
$\{v_1,w_1\} \neq \{0\}$ and $\{v_2,w_2\} \neq \{0\}$.
\end{proposition}
\begin{proof}
Assuming $v_2 w_1 w_2 \neq 0$, for a zero of $\det(\Omega)$ we would need
\bez
       e^{\varphi_1 - \varphi_2} = \frac{v_1 w_2}{v_2 w_1} \, , \qquad
          e^{\varphi_1 + \varphi_2} = - \frac{v_1 v_2^\ast}{w_1 w_2^\ast} \, ,
\eez
where $\varphi_k = \gamma_k t + \gamma_k^{-1} x$, $k=1,2$.  Hence 
\bez
           e^{2 \, \mathrm{Re}(\varphi_2)} = - \frac{|v_2|^2}{|w_2|^2} \, ,
\eez
which contradicts the positivity of the real exponential function.
If any of $v_2,w_1,w_2$ is zero, a zero of $\det(\Omega)$ is only possible if either
$v_1$ and $w_1$, or $v_2$ and $w_2$ are zero, which we excluded. 
\end{proof}

\begin{example}
\label{ex:2-soliton}
Choosing $n=2$, $\gamma_1=1$, $\gamma_2=2$ and  $v_1=v_2=w_1=w_2 =1$, (\ref{zero_seed_multi_soliton}) 
yields the special real 2-soliton solution
\bez
        f = 6 \, \frac{ \cosh(2 t + \frac{1}{2} x) - 2 \cosh(t+x) }{ \cosh(3 t + \frac{3}{2} x) + 9 \cosh(t - \frac{1}{2} x) -8 } \, .
\eez
A plot is shown in Fig.~\ref{fig:2soliton}.
\begin{figure}[h]
\begin{center}
\includegraphics[scale=.4]{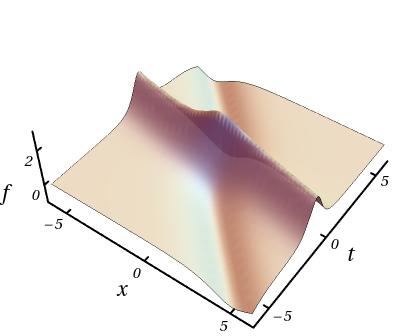} 
\parbox{15cm}{
\caption{Plot of the real 2-soliton solution in Example~\ref{ex:2-soliton}.
\label{fig:2soliton} } 
}
\end{center}
\end{figure}
\end{example}

\subsubsection{Solitons associated with non-diagonal Jordan matrices}
\label{subsec:zero_seed_Jordan}
In contrast to simple solitons, solutions determined by (\ref{zero_seed_multi_soliton}), with  
$\Gamma$ chosen as a non-diagonal Jordan matrix, depend also rationally on $x$ and $t$ (cf.  \cite{Manas96} 
for the NLS case).  

For $n=2$ and $\Gamma = \Gamma_{(2)}$, we have
\bez
     \eta_1 = e^{- \frac{1}{2 \gamma}(x + \gamma^2 t)} \left( \begin{array}{c} v_1 \\  \frac{v_1}{2 \gamma^2} (x - \gamma^2 t) + v_2
                       \end{array}\right) \, , \qquad
     \eta_2 = e^{ \frac{1}{2 \gamma}(x + \gamma^2 t)} \left( \begin{array}{c} w_1 \\  -\frac{w_1}{2 \gamma^2} (x - \gamma^2 t) + w_2
                       \end{array}\right) \, , 
\eez
Using Example~\ref{ex:n=2Jordan_Lyapunov_sol}, we obtain
\bez
      f &=& \frac{1}{4 \, \mathrm{Re}(\gamma)^3 \,  \det(\Omega)} \Big(  w_1^2 \, \Big[ \mathrm{Re}(\gamma) 
                \, v_1^\ast w_1^\ast (x/\gamma^{\ast 2} - t) - \mathrm{Re}(\gamma) \,  (v_1 w_2 - v_2 w_1)^\ast 
               + v_1^\ast w_1^\ast \Big] \, e^{x/\gamma + \gamma \, t} \\
    && - v_1^{\ast 2} \, \Big[ \mathrm{Re}(\gamma) \, v_1 w_1  (x/ \gamma^2 - t) 
         - \mathrm{Re}(\gamma) \,  (v_1 w_2 -v_2 w_1) - v_1 w_1  \Big] 
         \, e^{-x/\gamma^\ast - \gamma^\ast \, t} \Big)  \, , 
\eez
where 
\bez
         \det(\Omega) &=&  \frac{1}{16\, \mathrm{Re}(\gamma)^4} \, \Big( 
           \Big( | v_1|^2 \, \left| e^{-x/\gamma - \gamma \, t} \right| + |w_1|^2 \, \left| e^{x/\gamma + \gamma \, t} \right| \Big)^2 \\
      &&  + 4 \, \mathrm{Re}(\gamma)^2 \, \big| v_1 w_2 -v_2 w_1 - v_1 w_1 \, (x/\gamma^2 -t) \big|^2 \Big) \, .
\eez

\begin{example}
\label{ex:Jordan-soliton}
Choosing $\gamma = v_1=v_2=w_1=w_2 =1$, yields 
\bez
        f = 4 \, \frac{ \cosh(x+t) + (x - t) \, \sinh(x+t)}{1 + 2 (x-t)^2 + \cosh(2 (x+t)) } \, .
\eez
A plot is shown in Fig.~\ref{fig:Jordan-soliton}.
\begin{figure}[h]
\begin{center}
\includegraphics[scale=.4]{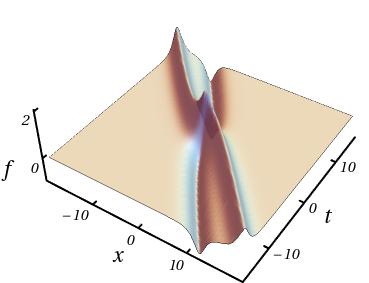} 
\parbox{15cm}{
\caption{Plot of the real solution in Example~\ref{ex:Jordan-soliton}.
\label{fig:Jordan-soliton} } 
}
\end{center}
\end{figure}
\end{example}

It is straightforward to work out solutions associated with Jordan matrices with $n>2$ and, more generally, solutions 
where $\Gamma$ is composed of different Jordan blocks, representing a nonlinear superposition of corresponding 
elementary solitons.

\subsection{Solutions with a plane wave background}
\label{subsec:pw_background}
Choosing as the ``seed" $f_0$ the plane wave solution
\be
          f_0 = C \, e^{\mathrm{i} \, (\alpha \, x - \beta \, t)} \, ,       \label{plane_wave_seed}
\ee
with a complex constant $C \neq 0$ and real constants $\alpha$ and $\beta$, (\ref{a,f_PDE}) determines 
$a_0= -\alpha \beta/2$. Writing
\be
     \eta_1 = e^{ -\frac{1}{2} \mathrm{i} \, (\alpha \, x - \beta \, t) } \, \tilde{\eta}_1 \, ,  \qquad
     \eta_2 = e^{ \frac{1}{2} \mathrm{i} \, (\alpha \, x - \beta \, t) } \, \tilde{\eta}_2 \, ,   \label{eta->teta}
\ee
the linear system, as rewritten in Proposition~\ref{prop:linsys_decoupled}, is converted into
\be
 &&  \tilde{\eta}_{1tt}  - ( \frac{1}{4} \tilde{\Gamma}^2 - |C|^2 ) \, \tilde{\eta}_1 = 0  \, ,  \qquad
   \mathrm{i} \, \alpha \, \tilde{\eta}_{1t} + \Gamma \, \tilde{\eta}_{1x} = 0 \, , 
                              \label{linsys_plane_wave} \\
 && \tilde{\eta}_2 = \frac{1}{C^\ast} \Big(  \tilde{\eta}_{1t} + \frac{1}{2} \tilde{\Gamma} \, \tilde{\eta}_1 \Big)  \, ,
                 \label{teta2}  
\ee
where
\bez
               \tilde{\Gamma} := \Gamma + \mathrm{i} \beta \, I_n \, .
\eez
The system is solved by
\be
   \tilde{\eta}_1 =  \cosh(\Theta) \, V \, , \qquad
   \tilde{\eta}_2 = \frac{1}{2 C^\ast} \, \Big(  \cosh(\Theta) \, \tilde{\Gamma} - 2 R \, \sinh(\Theta) \Big) V  \, ,
                 \label{plane_wave_linsys_sol}
\ee
with a constant $n$-component column vector $V$ and 
\bez
        \Theta := \mathrm{i} \, \alpha \, x \, \Gamma^{-1} R - t \, R + K 
                      = \mathrm{i} \, \alpha \, x \, (\tilde{\Gamma} - \mathrm{i} \beta \, I_n)^{-1} R - t \, R + K \,  .
\eez       
Here $R$ is a matrix root of
\be
            R^2 = \frac{1}{4} \tilde{\Gamma}^2 - |C|^2 I_n \, ,    \label{R_eq}
\ee
which is assumed to be invertible, 
and $K$ is a constant $n \times n$ matrix that commutes with $\Gamma$ (and then also with $R$).
The solution of (\ref{Lyap}) is given by 
\bez
  \Omega &=& \cosh(\Theta) \, X \, \cosh(\Theta^\dagger) 
          + \frac{1}{4 |C|^2} \cosh(\Theta) \,  \tilde{\Gamma} X  
           \tilde{\Gamma}^\dagger \cosh(\Theta^\dagger) 
          + \frac{1}{|C|^2} \sinh(\Theta) R X R^\dagger \sinh(\Theta^\dagger)\\
  &&  - \frac{1}{2 |C|^2} \Big( \sinh(\Theta) R X  \tilde{\Gamma}^\dagger  \cosh(\Theta^\dagger)
        + \cosh(\Theta) \tilde{\Gamma} X R^\dagger \sinh(\Theta^\dagger) \Big)  \, ,
\eez
where $X$ is the solution of the Lyapunov equation 
\be
        \tilde{\Gamma} \, X + X \, \tilde{\Gamma}^\dagger  =  \Gamma \, X + X \, \Gamma^\dagger = V V^\dagger \, .
                            \label{pw_Lyap}
\ee
(\ref{BDT_new_solution_f}) yields
\bez
        f = C \, e^{\mathrm{i} \, (\alpha \, x - \beta \, t)} \, \Big( 1 - \frac{1}{|C|^2} V^\dagger \cosh(\Theta^\dagger) \,
             \Omega^{-1} \, \big( \cosh(\Theta) \, \tilde{\Gamma} - 2 R \, \sinh(\Theta) \big) V  \Big) \, .
\eez
Essentially, for specified data, it remains to compute the inverse of $\Omega$ in order to 
find a more explicit form of this solution. 

\subsubsection{Simple solitons on a plane wave background}
If $\Gamma = \mathrm{diag}(\gamma_1,\ldots,\gamma_n)$, where $\gamma_i \neq -\gamma_j$, $i,j=1,\ldots,n$,
then $X$ is given by the Cauchy-like matrix
\bez
        X = \left( \frac{v_i \, v_j^\ast}{\gamma_i + \gamma_j^\ast} \right)  
            = \left( \frac{v_i \, v_j^\ast}{\tilde{\gamma}_i + \tilde{\gamma}_j^\ast} \right)   \, .
\eez

\begin{example}
For $n=1$, we find 
\bez
      \Omega &=& \frac{|V|^2}{8 \mathrm{Re}(\gamma)} \Big( 4 \,  |\cosh(\Theta)|^2 
    + \frac{1}{|C|^2} | \tilde{\gamma} \, \cosh(\Theta) - 2 r \, \sinh(\Theta) |^2 \Big) \, ,
\eez
where now $\Theta = (\mathrm{i} \, \alpha \, x /(\tilde{\gamma}-\mathrm{i} \beta)   - t ) r + K$ with 
$r = \pm \sqrt{\frac{1}{4} \tilde{\gamma}^2 - |C|^2}$.
(\ref{BDT_new_solution_f}) yields the solution
\be
       f = C \, e^{ \mathrm{i} \, (\alpha \, x - \beta \, t) }  \left( 1
       - \frac{ \mathrm{Re}(\tilde{\gamma}) \, \cosh(\Theta^\ast) 
             \big( \tilde{\gamma} \cosh(\Theta) - 2 \, r \sinh(\Theta) \big)}
       {  |C|^2 \, |\cosh(\Theta)|^2  + \frac{1}{4}  | \tilde{\gamma} \, \cosh(\Theta) - 2 \, r \, \sinh(\Theta) |^2  } 
      \right) \, .            \label{AKM_breather}
\ee
Fig.~\ref{fig:plane_wave_1soliton} shows a plot of the absolute value of $f$ for specified data. 
We note that $f = - f_0$ if $\tilde{\gamma} = 2 \, |C|$ (and thus $r=0$). 
Comparison of (\ref{AKM_breather}) with the focusing NLS solution (3.11) in \cite{CMH17} shows 
that, if $\tilde{\gamma}$ is real, we have a counterpart of a single Akhmediev breather if $|\tilde{\gamma}|<2|C|$ 
and a Kuznetsov-Ma breather if $|\tilde{\gamma}|>2|C|$. We refer to the references in \cite{CMH17} for the original 
literature on the NLS breathers. 
\begin{figure}[h]
\begin{center}
\includegraphics[scale=.4]{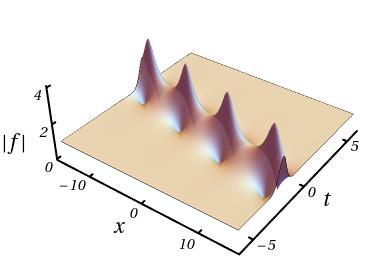} 
\hspace{2cm}
\includegraphics[scale=.4]{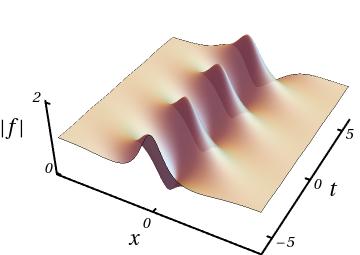} 
\parbox{15cm}{
\caption{Plots of the absolute value of $f$ for a single soliton on a plane wave background. Here we chose the data 
$\tilde{\gamma} =3, C=\alpha=\beta=v=1$ (real $r$, left plot), respectively 
$\tilde{\gamma} = 1, C=\alpha=\beta=v=1$ (imaginary $r$, right plot), and $K=0$.
 \label{fig:plane_wave_1soliton} } 
}
\end{center}
\end{figure}
\end{example}

For $n>1$, we obtain (nonlinear) superpositions of  the ``elementary solutions'', given in the preceding 
example.

\subsubsection{Solutions associated with non-diagonal Jordan matrices}

Following \cite{CMH17}, it is straightforward to elaborate the counterparts of the $n$-th order 
Akhmediev and Kuznetsov-Ma breathers. They are obtained by choosing $\tilde{\Gamma}$ to be 
an $n \times n$ Jordan matrix, i.e., (\ref{Gamma_Jordan_block}) with $k=n$ and $\gamma$ replaced by $\tilde{\gamma}$.\footnote{For the case of a $2 \times 2$ Jordan 
matrix $\Gamma$ and the NLS equation, also see \cite{Manas96}.}
The corresponding solution of the Lyapunov equation can be taken from Section~\ref{subsec:zero_seed_Jordan}.
A root of (\ref{R_eq}) is given by the Toeplitz matrix
\bez
   R = \frac{1}{2} \, \sqrt{\tilde{\gamma}^2 - 4 |C|^2} \, \left( \begin{array}{cccccc} 1 & 0 & \cdots & \cdots & \cdots & 0 \\ 
                                      \tilde{\gamma} \, (\tilde{\gamma}^2 - 4 |C|^2)^{-1}  &  1  & \ddots & \ddots & \ddots & \vdots \\
                                     - 2 \, (\tilde{\gamma}^2 - 4 |C|^2)^{-2} & \ddots  &  \ddots & \ddots & \ddots & \vdots \\
                              2 \, \tilde{\gamma} \, (\tilde{\gamma}^2- 4 |C|^2)^{-3} & \ddots & \ddots  &  \ddots &  \ddots & \vdots \\
                                 \vdots & \ddots & \ddots & \ddots & \ddots & 0 \\
                                 \vdots &  & \ddots & \ddots & \ddots & 1
                     \end{array} \right) \, ,
\eez
which commutes with the above Jordan matrix $\tilde{\Gamma}$, since this holds for any lower triangular 
$n \times n$ Toeplitz matrix. The entries in the first column of $R$ are the Taylor series coefficients of 
$\frac{1}{2} \, \sqrt{(\tilde{\gamma}+z)^2 - 4 |C|^2}$ at $z=0$.

The following result parallels that in Section~3.1.2 of \cite{CMH17}.

\begin{proposition}
\label{prop:regularity}
Let $\tilde{\Gamma}$ be an $n \times n$ (lower triangular) Jordan matrix with $\mathrm{Re}(\tilde{\gamma}) \neq 0$ and 
$\tilde{\gamma} \neq \pm 2 |C|$.   
Then the solution of (\ref{3rd-order_eq}), obtained from the solution (\ref{plane_wave_linsys_sol})
of the linear system, is regular if the first component of the vector $V$ 
is different from zero.  Furthermore, without restriction of generality one can set 
$V = (1,0,\ldots,0)^T$. \qed
\end{proposition}

\subsubsection{A degenerate case}
\label{subsec:degen}
The previously found solution of the linear system is only valid if the matrix $\frac{1}{4} \tilde{\Gamma}^2 - |C|^2 I_n$ 
is invertible. Now we drop this assumption. 

\begin{example}
\label{ex:Peregrine}
Let $n=1$ and $\tilde{\gamma}=2 |C|$. Then we have 
\bez
          \tilde{\eta}_{1tt}=0 \, , \qquad 
          \tilde{\eta}_1 = c_0 + c_1 \, \big( \frac{ \alpha \, x}{2 \mathrm{i} \, |C|+\beta}  + t \big)  \, , \qquad
          \tilde{\eta}_2 = \frac{C}{|C|} \, \tilde{\eta}_1 + \frac{c_1}{C^\ast}   \, ,
\eez
with complex constants $c_0,c_1$. The corresponding solution of (\ref{Lyap}) is
\bez
    \Omega &=& \frac{1}{4 |C|} ( |\tilde{\eta}_1|^2 + |\tilde{\eta}_2|^2) 
                          = \frac{1}{4 |C|} \Big( 2 \, \Big| \tilde{\eta}_1+ \frac{c_1}{2 |C|} \Big|^2 + \frac{|c_1|^2}{2 |C|^2} \Big) \\
&=& \frac{1}{2 |C|} \Big(  \Big| c_0 + c_1 \big( \frac{\alpha \, x}{2 \mathrm{i} \, |C| + \beta} + t + \frac{1}{2 |C|} \big)  \Big|^2  
         + \frac{|c_1|^2}{4 |C|^2} \Big) \, ,
\eez
and we obtain
\bez
 f = C \, e^{ \mathrm{i} \, (\alpha \, x - \beta \, t) } \Big[ 1
    - \frac{1}{|C| \, \Omega}  \Big(  \Big|c_0 + c_1 \, \big( \frac{ \alpha \, x}{2 \mathrm{i} \, |C|+\beta}  + t \big) \Big|^2 
    + \frac{c_1}{|C|} \big[ c_0^\ast + c_1^\ast \, \big( \frac{ \alpha \, x}{-2 \mathrm{i} \, |C|+\beta}  + t \big) \big] \Big) 
           \Big] \, .
\eez
This quasi-rational solution is the counterpart of the Peregrine breather solution of the focusing NLS equation, 
which models a rogue wave. 
Also see Fig.~\ref{fig:Peregrine}.
\begin{figure}[h]
\begin{center}
\includegraphics[scale=.4]{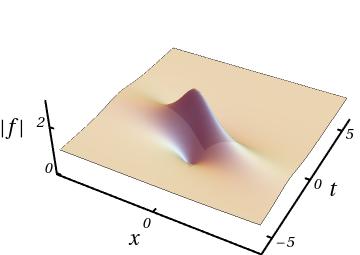} 
\parbox{15cm}{
\caption{Plot of the absolute value of the solution in Example~\ref{ex:Peregrine}. Here we chose
the data $\alpha=\beta=C=c_0=c_1=1$.
\label{fig:Peregrine} } 
}
\end{center}
\end{figure}
\end{example}

Counterparts of higher order Peregrine breathers are obtained if $\tilde{\Gamma}$ is a Jordan matrix with eigenvalue 
$\tilde{\gamma} = 2 |C|$.\footnote{The second order Peregrine breather solution of the NLS equation 
appeared first in \cite{Manas96} to our knowledge, see Section 4.1.2 therein.} 
We follow the steps in Section~3.2 of \cite{CMH17} and omit the analogous proofs.

\begin{proposition}
\label{prop:pw_linsys_sol}
Let 
\bez
           \mathcal{N}  := \frac{1}{4} \tilde{\Gamma}^2 - |C|^2 I_n
\eez 
be nilpotent of degree $N>0$. (\ref{linsys_plane_wave}) is then solved by
\bez
        \tilde{\eta}_1 &=& \big( R_1(\mathcal{N},t) \, R_1(-\alpha^2 \Gamma^{-2} \mathcal{N},x)
           - \mathrm{i}  \,  \alpha \, \Gamma^{-1}  \mathcal{N} \, R_2(\mathcal{N},t) \, R_2(-\alpha^2 \Gamma^{-2} \mathcal{N},x)\big) \, v \\  
   && + \big( R_1(\mathcal{N},t) \,  R_2(-\alpha^2 \Gamma^{-2} \mathcal{N},x) 
          + \frac{\mathrm{i}}{\alpha} \Gamma \, R_2(\mathcal{N},t) \,  R_1(-\alpha^2 \Gamma^{-2} \mathcal{N},x) \big) \, w \, , 
\eez
with constant $n$-component vectors $v,w$, and 
\bez
     R_1(\mathcal{N},t) := \sum_{k=0}^{N-1} \frac{t^{2k}}{(2k)!} \, \mathcal{N}^k \, , \qquad
     R_2(\mathcal{N},t) := \sum_{k=0}^{N-1} \frac{t^{2k+1}}{(2k+1)!} \, \mathcal{N}^k \, .
\eez
$\tilde{\eta}_2$ is given by (\ref{teta2}).   \qed
\end{proposition}

Again, the corresponding solution of the Lyapunov equation is obtained via the results recalled in  Section~\ref{subsec:zero_seed_Jordan}.

\begin{proposition}
Let $\tilde{\Gamma}$ be a (lower triangular) Jordan matrix with eigenvalue 
$\tilde{\gamma} = 2 |C|$. Let $\eta_1,\eta_2$ be the solution of the linear system,  given by Proposition~\ref{prop:pw_linsys_sol} and (\ref{eta->teta}). The solution of  (\ref {3rd-order_eq}),  
obtained via (\ref{BDT_new_solution_f}), is then regular if the first component of $v$ is different from zero. 
Moreover, without restriction of generality, we can set $v=(1,0,\ldots,0)^T$. \hfill   $\qed$
\end{proposition}

\section{Derivation of the binary Darboux transformation}
\label{sec:bDT}

\subsection{Binary Darboux transformations in bidifferential calculus}
\label{subsec:bDT_in_bdc}
A \emph{graded associative algebra} is an associative algebra $\boldsymbol{\Omega} = \bigoplus_{r \geq 0} \boldsymbol{\Omega}^r$
over a field $\mathbb{K}$ of characteristic zero, where $\mathcal{A} := \boldsymbol{\Omega}^0$ is an associative 
algebra over $\mathbb{K}$ and $\boldsymbol{\Omega}^r$, $r \geq 1$, are $\mathcal{A}$-bimodules such that 
$\boldsymbol{\Omega}^r \, \boldsymbol{\Omega}^s \subseteq \boldsymbol{\Omega}^{r+s}$. Elements of $\boldsymbol{\Omega}^r$ 
will be called $r$-forms.
A \emph{bidifferential calculus} is a unital graded associative algebra $\boldsymbol{\Omega}$, supplied
with two $\mathbb{K}$-linear graded derivations 
$\d, \bd : \boldsymbol{\Omega} \rightarrow \boldsymbol{\Omega}$
of degree one (hence $\d \boldsymbol{\Omega}^r \subseteq \boldsymbol{\Omega}^{r+1}$,
$\bd \boldsymbol{\Omega}^r \subseteq \boldsymbol{\Omega}^{r+1}$), and such that
\begin{eqnarray}
    \d^2 = \bd^2 = \d \bd + \bd \d = 0  \, .   
    \label{bidiff_conds}
\end{eqnarray}
We refer the reader to \cite{DMH20dc} for an introduction to this structure and an extensive list of references. 

\begin{theorem}
\label{thm:A}
Given a bidifferential calculus, let 0-forms $\Delta, \Gamma$ and 1-forms  $\kappa,\lambda$ satisfy  
\be
  &&  \bd \Delta + [\lambda , \Delta] = (\d \Delta) \, \Delta \, , \qquad  
      \bd \lambda + \lambda^2 = (\d \lambda) \, \Delta \, , \nonumber \\ 
  &&  \bd \Gamma - [\kappa , \Gamma] = \Gamma \, \d \Gamma  \, , \qquad  
      \bd \kappa - \kappa^2 = \Gamma \, \d \kappa \, .   \label{Delta,lambda,Gamma,kappa_eqs}
\ee
Let 0-forms $\theta$ and $\eta$ be solutions of the linear equations 
\be
    \bd \theta = A \, \theta + (\d \theta) \, \Delta + \theta \, \lambda \, , \qquad
   \bd \eta = - \eta \, A + \Gamma \, \d \eta + \kappa \, \eta \, ,   \label{linsys1}
\ee
where the 1-form $A$ satisfies\footnote{Under suitable assumptions for $\Delta$ and $\Gamma$, 
these equations arise as integrability conditions of the linear system and ``adjoint linear system'' given in 
(\ref{linsys1}), by use of (\ref{Delta,lambda,Gamma,kappa_eqs}). In any case, the integrability conditions 
are satisfied if (\ref{Delta,lambda,Gamma,kappa_eqs}) and (\ref{A_eqs}) hold.}
  \be
         \d A = 0 \, , \qquad \bd A = A^2 \, .   \label{A_eqs}
\ee
Furthermore, let $\Omega$ be an invertible solution of the linear system\footnote{The equation obtained by 
acting with $\bd$ on (\ref{linsys2}) is satisfied as a consequence of the preceding equations.}
\be
    && \Gamma \, \Omega - \Omega \, \Delta = \eta \, \theta \, ,  \label{genSylv} \\
    && \bd \Omega = (\d \Omega) \, \Delta - (\d \Gamma) \, \Omega + \kappa \, \Omega + \Omega \, \lambda + (\d \eta) \, \theta \, .
        \label{linsys2}
\ee
Then  
\be
        A' := A - \d (\theta \, \Omega^{-1} \eta)      \label{A'}
\ee
also solves (\ref{A_eqs}). 
\end{theorem}
\begin{proof}
Clearly, we have $\d A' =0$. Using (\ref{A_eqs}), we obtain
\bez
    \bd A' - A'^2 = \d \bd (\theta \, \Omega^{-1} \eta) + A \, \d (\theta \, \Omega^{-1} \eta) 
    + \theta \, \Omega^{-1} \eta \, A - \d (\theta \, \Omega^{-1} \eta) \, \d (\theta \, \Omega^{-1} \eta) \, .
\eez
With the help of the linear equations (\ref{linsys1}) and (\ref{linsys2}), we find
\bez
      \bd (\theta \, \Omega^{-1} \eta)  &=& A  \, \theta \, \Omega^{-1} \eta - \theta \, \Omega^{-1} \eta \, A 
    + (\d \theta) \, \Delta \Omega^{-1} \eta + \theta \, \Omega^{-1} \Gamma \, \d \eta \\
    && - \theta \, \Omega^{-1} (\d \Omega) \, \Delta \, \Omega^{-1} \eta 
    + \theta \, \Omega^{-1} (\d \Gamma) \, \eta - \theta \, \Omega^{-1} (\d \eta) \, \theta \, \Omega^{-1} \eta \, .
\eez
Eliminating $\Gamma$ using (\ref{genSylv}), it becomes
\bez
        \bd (\theta \, \Omega^{-1} \eta)  = A  \, \theta \, \Omega^{-1} \eta - \theta \, \Omega^{-1} \eta \, A
       + \d (\theta \Delta \, \Omega^{-1} \eta) + \theta \, \Omega^{-1} \eta \, \d (\theta \, \Omega^{-1} \eta) \, .
\eez
Inserting this in our first equation leads to $\bd A' - A'^2 = 0$.
\end{proof}

The preceding theorem, and also the result stated next, remain true if the ingredients are \emph{matrices} of forms with 
dimensions chosen in such a way that the required products are all defined. Furthermore, it will be sufficient to have the 
maps $\d$ and $\bd$ defined on those matrices that appear in the theorem, but not necessarily on the whole of 
$\boldsymbol{\Omega}$.

\begin{corollary}
\label{cor:phi}
Let (\ref{Delta,lambda,Gamma,kappa_eqs}) hold and (\ref{linsys1}) with $A=\d \phi$, where the 0-form $\phi$ 
is a solution of
\be
         \bd \d \phi = \d \phi \, \d \phi \, .     \label{phi_eq}
\ee
If $\Omega$ is an invertible solution of (\ref{genSylv}) and (\ref{linsys2}), then  
\be
        \phi' = \phi - \theta \, \Omega^{-1} \eta + K \, ,
\ee
where $K$ is any $\d$-constant (i.e., $\d K=0$), solves the same equation.
\end{corollary}

The result in Corollary \ref{cor:phi} can be regarded as a reduction  
of that in Theorem \ref{thm:A}. Corollary~\ref{cor:phi} has been used in many previous applications of 
bidifferential calculus, see in particular \cite{DMH20dc,DMH13SIGMA,CDMH16}. 
Here we provided short proofs of the above general results. Below, we will use Corollary \ref{cor:phi} to deduce
Theorem~\ref{thm:f_eq}.

\subsection{An application}
Let $\mathcal{A}$ be a unital associative algebra over $\mathbb{C}$, where the elements are allowed to depend 
on real variables $x,t$.  Let $\mathrm{Mat}(\mathcal{A})$ be the algebra of all matrices over $\mathcal{A}$, 
where the product of two matrices is defined to be zero whenever their dimensions do not fit. 
We choose
\begin{eqnarray}
    \boldsymbol{\Omega} = \mathrm{Mat}(\mathcal{A}) \otimes \bigwedge \mathbb{C}^2 \, ,   
                                                                   \label{Omega_wedge}
\end{eqnarray}
where $\bigwedge \mathbb{C}^2$ is the exterior algebra of the vector space $\mathbb{C}^2$. 
It is then sufficient to define $\d$ and $\bd$ on $\mathrm{Mat}(\mathcal{A})$, since they extend 
in an evident way to $\boldsymbol{\Omega}$, treating elements of $\bigwedge \mathbb{C}^2$ as $\d$- and $\bd$-constants. 

Let $\xi_1,\xi_2$ be a basis of $\bigwedge^1 \mathbb{C}^2$. 
For each $m \in \mathbb{N}$, let $J_m$ be a constant $m \times m$ matrix over $\mathcal{A}$.  
For an $m \times n$ matrix $F$ over $\mathcal{A}$, let
\bez
       \d F = F_x \, \xi_1 + \frac{1}{2} (J_m F - F J_n) \, \xi_2 \, , \qquad
      \bd F = \frac{1}{2} (J_m F - F J_n) \, \xi_1 + F_t \, \xi_2 
\eez
(also see \cite{DMH10AKNS,DKMH11acta,DMH10NLS}).\footnote{On a $1 \times 1$ matrix $f$, which is an 
element of $\mathcal{A}$, we have $\d f = f_x \, \xi_1$ and $\bd f = f_t \, \xi_2$.}
Then $\d$ and $\bd$ satisfy the Leibniz rule on a 
product of matrices, and the conditions in (\ref{bidiff_conds}) are satisfied.
In the linear systems (\ref{linsys1}) we choose a $2 \times n$ matrix $\theta$ and an $n \times 2$ matrix $\eta$. 
Then $A$ has to be a $2 \times 2$ matrix of 1-forms. Writing
\bez
         A = A_1 \xi_1 + A_2 \xi_2 \, , \qquad
         \kappa = \kappa_1 \xi_1 + \kappa_2 \xi_2 \, , \qquad
        \lambda = \lambda_1 \xi_1 + \lambda_2 \xi_2 \, ,
\eez 
with $2 \times 2$ matrices (over $\mathcal{A}$) $A_1$ and $A_2$. (\ref{linsys1}) reads
\bez
     &&   \frac{1}{2} (J_2 \theta - \theta J_n) = A_1 \theta + \theta_x \Delta + \theta \lambda_1 \, , \qquad
              \theta_t = A_2 \theta + \frac{1}{2} (J_2 \theta - \theta J_n) \Delta + \theta \lambda_2 \, , \\
     &&   \frac{1}{2} (J_n \eta - \eta J_2) = - \eta A_1 + \Gamma \eta_x + \kappa_1 \eta  \, , \qquad
              \eta_t = - \eta A_2 + \frac{1}{2} \Gamma (J_n \eta - \eta J_2) + \kappa_2 \eta  \, . 
\eez
Choosing 
\bez
            \kappa_1 = \frac{1}{2} J_n \, , \qquad
           \kappa_2 = - \frac{1}{2} \Gamma J_n \, , \qquad
          \lambda_1 = - \frac{1}{2} J_n \, , \qquad
         \lambda_2 = \frac{1}{2} J_n \Delta \, , 
\eez
the latter system simplifies to
\bez
      &&   \frac{1}{2} J_2 \theta = A_1 \theta + \theta_x \Delta  \, , \qquad 
              \theta_t = A_2 \theta + \frac{1}{2} J_2 \theta \Delta  \, , \\
     &&   \frac{1}{2} \eta J_2 =  \eta A_1 - \Gamma \eta_x \, , \qquad
             \eta_t = - \eta A_2 - \frac{1}{2} \Gamma \eta J_2 \, ,
\eez
which does not involve $J_n$ with $n \neq 2$ anymore. 
The conditions in (\ref{Delta,lambda,Gamma,kappa_eqs}) boil down to
\bez
        \Delta_x = \Delta_t = 0 \, , \qquad \Gamma_x = \Gamma_t = 0 \, ,
\eez
so that $\Delta$ and $\Gamma$, which are $n \times n$ matrices over $\mathcal{A}$, have to be constant.
(\ref{linsys2}) becomes
\bez
    \Omega_x \Delta = - \eta_x \theta \, , \qquad   \Omega_t = - \frac{1}{2} \eta J_2 \theta \, .
\eez
In addition, the $n \times n$ matrix $\Omega$ has to satisfy (\ref{genSylv}). 
Choosing $J_2 = \mathrm{diag}(\mathbf{1},-\mathbf{1})$, where $\mathbf{1}$ stands for the identity 
element of $\mathcal{A}$, and writing
\bez
       \phi = \left( \begin{array}{cc} p & f \\ q & -\tilde{p} \end{array} \right) \, , 
\eez
elaborating Corollary~\ref{cor:phi}, we have
\bez
            A = \d \phi = \left(\begin{array}{cc} p_x & f_x \\ q_x & -\tilde{p}_x \end{array}\right) \, \xi_1 
                + \left(\begin{array}{cc} 0 & f \\ -q & 0 \end{array}\right) \, \xi_2 \, ,
\eez
and (\ref{phi_eq}) takes the form
\bez
         \phi_{xt} = \frac{1}{2} \Big[ [J,\phi],\phi_x - \frac{1}{2} J \Big] \, ,
\eez
also see \cite{DKMH11acta}. This results in the system
\bez    
       f_{xt} = f - p_x f - f \tilde{p}_x \, , \qquad q_{xt} = q - q p_x - \tilde{p}_x q  \, , \qquad
       p_{xt} = (f q)_x \, , \qquad \tilde{p}_{xt} = (q f)_x   \, .
\eez
Introducing
\bez
          a := p_x - \frac{1}{2} \, \mathbf{1} \, , \qquad \tilde{a} := \tilde{p}_x - \frac{1}{2} \, \mathbf{1}  \, ,
\eez
it reads
\be
       f_{xt} = - a \, f - f \, \tilde{a} \, , \qquad q_{xt} =  - q \, a - \tilde{a} \, q  \, , \qquad
       a_t = (f q)_x \, , \qquad \tilde{a}_t = (q f)_x   \, .    \label{nc_f,a-system}
\ee

The constraint 
\bez
               q = \pm f^\dagger \, , \qquad a^\dagger = a \, , \qquad \tilde{a}^\dagger = \tilde{a} \, , 
\eez
reduces the last system to
\bez
       f_{xt} = - a \, f - f \, \tilde{a} \, , \qquad 
       a_t = \pm (f f^\dagger)_x  \, , \qquad \tilde{a}_t = \pm (f^\dagger f)_x \, .
\eez

\begin{remark}
Instead, the reduction $q=\pm f, \tilde{a} = a$ leads to
\bez
           f_{xt} = - a \, f - f \, a    \, , \qquad     a_t = \pm (f^2)_x      \, .
\eez
Choosing the upper sign, this system may be regarded, as has been suggested in \cite{DKMH11acta} (see equation (9) therein), 
as a matrix version of the self-induced transparency (SIT) equations.
\end{remark}

\subsection{The commutative case}
Let $\mathcal{A}$ now be the commutative algebra of functions on $\mathbb{R}^2$. 
Then we have 
\bez
      \mathrm{tr}(\theta \Omega^{-1} \eta) = \mathrm{tr}(\eta \theta \Omega^{-1}) 
   = \mathrm{tr}( (\Gamma \Omega - \Omega \Delta) \, \Omega^{-1}) = \mathrm{tr}( \Gamma ) - \mathrm{tr}( \Delta ) \, ,
\eez
where $\Gamma$ and $\Delta$ shall now be $n \times n$ matrices over $\mathbb{C}$, $\theta$ and $\eta$ of size 
$k \times n$ and $n \times k$, respectively. 
Hence
\bez
          \mathrm{tr}( \d (\theta \Omega^{-1} \eta)) = \mathrm{tr}( \theta \Omega^{-1} \eta)_x \, \xi_1 
          = (\mathrm{tr}( \Gamma ) - \mathrm{tr}( \Delta ))_x  \, \xi_1 = 0 \, .
\eez
As a consequence,  
\bez
            \mathrm{tr} A = 0
\eez
is a reduction that is consistent with the solution-generating method of Theorem \ref{thm:A}. 
The latter reduction means $\tilde{a} = a$, so that (\ref{nc_f,a-system}) becomes
\be
        a_t = (f q)_x  \, , \qquad  
        f_{xt} =  - 2 a \, f   \, , \qquad      
        q_{xt} = -2 a \, q     \, .  \label{f,q,a-system}
\ee
This is the first ``negative flow" of the AKNS hierarchy (see, e.g., system (12) in \cite{Kamc+Pavl02} 
and also  \cite{DMH10AKNS,Veks12}, for example).  
Soliton solutions of it have, apparently, first been found in \cite{JZZ09}, using Hirota's bilinear method. 

Since we guaranteed form-invariance of $A$ under the transformation given by Corollary~\ref{cor:phi}, 
writing 
\bez 
         \theta = \left(  \begin{array}{c} \theta_1 \\ \theta_2 \end{array}  \right) \, , \quad  \eta = (\eta_1,\eta_2) \, ,
\eez
we have
\bez
        \d \phi' &=&  A' = A - \d (\theta \, \Omega^{-1} \eta) 
            = \left(\begin{array}{cc} a' + \frac{1}{2} & f'_x \\  q'_x & - a' - \frac{1}{2}) \end{array}\right) \, \xi_1 
                + \left(\begin{array}{cc} 0 & f' \\ - q' & 0 \end{array}\right) \, \xi_2  \\
         &=&  \left(\begin{array}{cc} a + \frac{1}{2} - (\theta_1 \Omega^{-1} \eta_1)_x & (f- \theta_1 \Omega^{-1} \eta_2)_x \\  
                  (q - \theta_2 \Omega^{-1} \eta_1)_x & -a - \frac{1}{2} - (\theta_2 \Omega^{-1} \eta_2)_x \end{array}\right) \, \xi_1 
                + \left(\begin{array}{cc} 0 & f - \theta_1 \Omega^{-1} \eta_2 \\ -q + \theta_2 \Omega^{-1} \eta_1 & 0 \end{array} \right) \, \xi_2 
                   \,  ,
\eez
and hence
\bez
          a' = a - (\theta_1 \Omega^{-1} \eta_1)_x = a + (\theta_2 \Omega^{-1} \eta_2)_x \, , \qquad
          f' = f - \theta_1 \Omega^{-1} \eta_2 \, , \qquad
         q' = q - \theta_2 \Omega^{-1} \eta_1  \, ,
\eez
satisfy the same equations as $a,f,q$. 
Collecting the main results, we arrive at the following theorem, which expresses a binary Darboux transformation
for the system (\ref{f,q,a-system}).

\begin{theorem}
\label{thm:fq-system}
Let $a_0, f_0, q_0$ be a solution of (\ref{f,q,a-system}). 
Let $\theta_i$ and $\eta_i$, $i=1,2$, be solutions of the linear system
\bez
      &&   \theta_{1x} \Delta = - a_0 \, \theta_1 - f_{0x} \, \theta_2 \, , \qquad
         \theta_{2x} \Delta = a_0 \, \theta_2 - q_{0x} \, \theta_1 \, , \nonumber \\  
       &&  \theta_{1t} = \frac{1}{2} \theta_1 \Delta + f_0 \, \theta_2  \, , \qquad
         \theta_{2t} = - \frac{1}{2} \theta_2 \Delta - q_0 \, \theta_1  \, , \nonumber \\
       && \Gamma \, \eta_{1x} = a_0 \, \eta_1 + q_{0x} \, \eta_2  \, , \qquad
        \Gamma \, \eta_{2x} =  - a_0 \, \eta_2 + f_{0x} \, \eta_1   \, , \nonumber \\
      && \eta_{1t} = - \frac{1}{2} \Gamma \, \eta_1 + q_0  \, \eta_2  \, , \qquad
       \eta_{2t} = \frac{1}{2} \Gamma \, \eta_2 - f_0 \, \eta_1  \, , 
\eez
where $\Delta$ and $\Gamma$ are invertible constant $n \times n$ matrices. 
Let $\Omega$ be an invertible solution of the linear equations
\be
    &&     \Gamma \, \Omega - \Omega \, \Delta = \eta_1 \theta_1 + \eta_2 \theta_2 \, , \label{Sylv} \\
    &&    \Omega_x \Delta = - \eta_{1x} \theta_1 - \eta_{2x} \theta_2 \, , \qquad
            \Omega_t = - \frac{1}{2} \eta_1 \theta_1 +  \frac{1}{2} \eta_2 \theta_2 \, .  \label{Omega_diff_eqs}
\ee
Then 
\be
         a = a_0 - (\theta_1 \Omega^{-1} \eta_1)_x \, , \qquad
         f = f_0 - \theta_1 \Omega^{-1} \eta_2 \, , \qquad
         q = q_0 - \theta_2 \Omega^{-1} \eta_1  \, ,                              \label{f,q_new_solutions}
\ee
constitutes also a solution of (\ref{f,q,a-system}). \hfill $\qed$
\end{theorem}

If we impose the condition
\be
              q = \pm f^\ast \, ,     \label{reduction}
\ee
the system (\ref{f,q,a-system}) reduces to (\ref{a,f_PDE}), if we choose the plus sign.
It remains to implement the above reduction in the solution-generating method.

\subsection{The reduction $q = f^\ast$}
Let us set
\be
      q = f^\ast \, , \qquad   \theta = \eta^\dagger \, , \qquad \Delta = - \Gamma^\dagger \, .
                \label{conj_red-}
\ee
Then Theorem~\ref{thm:fq-system} implies Theorem~\ref{thm:f_eq}. 

\paragraph{Proof of Theorem~\ref{thm:f_eq}.}  The linear system in Theorem~\ref{thm:fq-system} 
reduces to (\ref{eta_x}) and (\ref{eta_t}), by using (\ref{conj_red-}), and (\ref{Sylv}) becomes (\ref{Lyap}). 
If $\Gamma$ and $-\Gamma^\dagger$ have no eigenvalue in common, i.e., 
$\mathrm{spec}(\Gamma) \cap \mathrm{spec}(-\Gamma^\dagger) = \emptyset$, 
the Lyapunov equation (\ref{Lyap}) is known to have a unique solution $\Omega$. By taking its conjugate, we can 
then deduce that 
\be
              \Omega^\dagger = \Omega \, ,    \label{Omega_Hermitian}
\ee
which in turn implies that the equations 
\bez
   \Omega_x \Gamma^\dagger = \eta_{1x} \eta_1^\dagger + \eta_{2x} \eta_2^\dagger \, , \qquad
   \Omega_t = - \frac{1}{2} \eta_1 \eta_1^\dagger +  \frac{1}{2} \eta_2 \eta_2^\dagger \, ,   \label{Omega_eqs}
\eez
resulting from (\ref{Omega_diff_eqs}), are satisfied as a consequence of the equations 
\bez
         \Omega_x \Gamma^\dagger - \eta_x \eta^\dagger + \Gamma \Omega_x - \eta \, \eta_x^\dagger = 0 \, , \qquad
        \Gamma \, (\Omega_t + \frac{1}{2} \eta_1 \eta_1^\dagger -  \frac{1}{2} \eta_2 \eta_2^\dagger)
        + (\Omega_t + \frac{1}{2} \eta_1 \eta_1^\dagger -  \frac{1}{2} \eta_2 \eta_2^\dagger) \, \Gamma^\dagger = 0 \, ,
\eez
obtained by differentiation of (\ref{Lyap}) with respect to $x$, respectively $t$, and using (\ref{eta_t}). 
Furthermore,
\bez  
        (\eta^\dagger \, \Omega^{-1} \eta)^\dagger =  \eta^\dagger \, \Omega^{-1} \eta \, , 
\eez
so that
\bez
        (\eta_1^\dagger \, \Omega^{-1} \eta_2)^\ast = \eta_2^\dagger \, \Omega^{-1} \eta_1 \, ,
\eez
and the last two equations in (\ref{f,q_new_solutions}) indeed coincide if (\ref{conj_red-}) holds.  
Furthermore, $\eta_1^\dagger \, \Omega^{-1} \eta_1$ is real, so that $a$ is real if $a_0$ is real, which is required 
by the linear system.
\hfill $\qed$

\section{Conclusion}
\label{sec:conclusion}
In this work we explored the nonlinear PDE (\ref{3rd-order_eq}), which is completely integrable (in the sense that a 
Lax pair exists) if the dependent variable is \emph{real}. Expressing (\ref{3rd-order_eq}) as the equivalent system (\ref{a,f_PDE}),  
in the \emph{complex} case integrability apparently requires that the function $a$ has to be real \cite{Sako22}. Accordingly, 
the vectorial binary Darboux transformation, which we generated  from 
a universal binary Darboux transformation in bidifferential calculus, only works for (\ref{a,f_PDE}) with the 
restriction to \emph{real} $a$.

We exploited the vectorial binary Darboux transformation to obtain multi-soliton solutions of (\ref{3rd-order_eq}), 
also on a plane wave background solution. This includes counterparts of the Akhmediev and Kuznetsov-Ma breathers and 
rogue wave solutions of the NLS equation. 
All these solutions can still be generalized by using Proposition~\ref{prop:gencov}. 

Little is known about the non-integrable sector of (\ref{3rd-order_eq}), which 
is  (\ref{a,f_PDE}) if the function $a$ has a non-zero imaginary part. In this sector, the only exact solution we know 
is the fairly trivial one  $f = C \, e^{\mathrm{i} \, \alpha(x) - \beta \, t}$, where $\alpha$ is any real function with 
non-vanishing derivative, $\beta \neq 0$ is a real constant and $C \neq 0$ a complex constant.  
That switching on an imaginary part of the function $a$ in (\ref{a,f_PDE})  apparently breaks complete integrability,  
is an interesting observation \cite{Sako22}, which deserves further exploration.

 More generally, we derived a vectorial binary Darboux transformation for the system (\ref{f,q,a-system}), 
which is the first ``negative flow" of the AKNS hierarchy. 

Changing the plus sign in (\ref{3rd-order_eq}) to a minus sign, leads us to a ``defocusing" version of it, using the familiar 
terminology for the corresponding versions of the NLS equation.  In this case we have to deal with the reduction 
condition (\ref{reduction}) with the choice of the minus sign and implement it in the binary Darboux transformation 
for the system (\ref{f,q,a-system}). This has not been elaborated in this work.

Most likely, the simple multi-soliton solutions admit generalizations in terms of Jacobi elliptic functions. 
We have seen that there are even \emph{two} such extensions of the 1-soliton solution. 

The plots in this work have been generated using {\sc Mathematica} \cite{Mathematica}. 

\vspace{.2cm}
\noindent
\textbf{Acknowledgment.} I would like to thank Rusuo Ye for reviving my interest in the bidifferential 
calculus in \cite{DMH10AKNS,DKMH11acta} and for an interesting communication. 
I'm greatly indebted to Maxim Pavlov for reminding me of the origin of the PDE (\ref{3rd-order_eq}) 
and for further help.

\end{document}